\documentclass[11pt]{article}

\usepackage[colorlinks=true, linkcolor=blue, citecolor=blue, urlcolor=blue, pagebackref]{hyperref}
\usepackage[dvipsnames,usenames]{color}
\usepackage{xcolor}

\hypersetup{
  linkcolor=[rgb]{0.3,0.3,0.6},
  citecolor=[rgb]{0.2, 0.6, 0.2},
  urlcolor=[rgb]{0.6, 0.2, 0.2}
}

\usepackage{fullpage}
\usepackage{amsmath}
\usepackage{amssymb}
\usepackage{amsthm}
\usepackage{thmtools,thm-restate}

\declaretheoremstyle[bodyfont=\it,qed=\qedsymbol]{noproofstyle}

\declaretheorem[numberlike=equation]{lemma}

\declaretheorem[numberlike=equation]{claim}
\declaretheorem[numberlike=equation,style=defstyle]{example}
\declaretheorem[numberlike=equation,style=defstyle]{remark}

\usepackage[nameinlink,capitalise]{cleveref}

\usepackage{nth}

\newcommand{\iv}[1]{\ensuremath{[\![{#1}]\!]}}
\newcommand{\ehref}[1]{\href{mailto:#1}{#1}}
\newcommand{\vp}{\ensuremath{\tau}}
\newcommand{\up}{\ensuremath{\sigma}}
\newcommand{\cp}{\ensuremath{\rho}}
\newcommand{\conv}{\text{conv}}
\newcommand{\cone}{\text{cone}}

\title{Optimal Union Probability Interval Is NP-Hard}

\author{ 
{Petteri Kaski\thanks{Department of Computer Science, Aalto University, Finland. \ehref{petteri.kaski@aalto.fi}.}}
\and 
{Heikki Mannila\thanks{Department of Computer Science, Aalto University, Finland. \ehref{heikki.mannila@aalto.fi}.}}
\and
{Chandra Kanta Mohapatra\thanks{Department of Computer Science, Aalto University, Finland. \ehref{chandra.mohapatra@aalto.fi}.}}
}

\date{}

\begin{document}

\maketitle

\begin{abstract}
A problem dating back to Boole~[{\em Laws of Thought}, Walton \& Maberly,~1854]
is what can be computed about the probability of a finite union of events 
when given as input the probabilities of intersections of some of the events.
The modern geometric study of the problem can be traced back to 
Hailperin [{\em Amer.~Math.~Monthly} 2 (1965) 343--359] who phrased the 
problem in the language of linear programming and generalized it to
logical formulas of the events other than disjunction, 
heralding a substantial body of work in probabilistic logic 
[Nilsson, {\em Artif.\ Intell.}\ 28~(1986) 71--87], including the 
{\em probabilistic satisfiability problem} of 
Georgakopoulos, Kavvadis, and Papadimitriou 
[{\em J.~Complexity} 4 (1988) 1--11], as well as fundamental connections to 
the geometry of metrics via cut and correlation polytopes 
[Deza and Laurent, {\em Geometry of Cuts and Metrics}, Springer, 1997]
and to the study of marginal polytopes in graphical models of machine learning 
[Wainwright and Jordan, {\em Found.\ Trends Mach.\ Learn.}\ 1 (2008) 1--305].
This paper (i)~describes the pertinent geometry of Boole's problem 
via coordinate projections of an elementary polytope arising essentially 
from Hailperin's linear program on the atoms of a Venn diagram, 
and (ii)~shows that computing the optimal interval for the union probability 
is NP-hard, resolving an apparent gap in the literature highlighted by Pitowsky [{\em Math.\ Programming} 50 (1991) 395--414] and 
Boros {\em et al.}\ [{\em Math.\ Oper.\ Res.}\ 39 (2014) 1311--1329 and 51 (2026) 134--148]. 
\end{abstract}

\section{Introduction}

\label{sect:introduction}

Suppose we have $n$ events $B_1,B_2,\ldots,B_n$ in a probability space
$\Omega$ and write $|B|$ for the probability of an event $B$.
The classical inclusion-exclusion formula 
\begin{equation}
\label{eq:inclusion-exclusion}
|\!\cup_{i\in[n]}B_i|\ =\!\sum_{\emptyset\neq S\subseteq [n]}(-1)^{|S|+1}|\!\cap_{i\in S}B_i|
\end{equation}
efficiently solves for the probability of the union event $\cup_{i\in[n]}B_i$, 
given as input the probabilities of all the intersection events 
$B_S:=\cap_{i\in S}B_i$ 
for all $\emptyset\neq S\subseteq[n]:=\{1,2,\ldots,n\}$ and 
$B_\emptyset:=\Omega$. But what can one compute about the union 
probability when one is given as input the probabilities of only some of the 
intersection events? More specifically, suppose we are given as input
a set family $\emptyset\neq\mathcal{F}\subseteq 2^{[n]}\setminus\{\emptyset\}$
together with the intersection probability $|B_S|$ for each $S\in\mathcal{F}$, 
what can we say about $|\!\cup_{i\in[n]}B_i|$?

This problem originates to Boole~\cite{Boole1854}, who 
observed that when only individual event probabilities are given, 
the union probability satisfies 
the inequalities
\begin{equation}
\label{eq:boole-frechet}
\max_{i\in[n]}|B_i|
\leq
|\!\cup_{i\in[n]}B_i|
\leq
\min\bigl(1,\sum_{i\in[n]}|B_i|\bigr)\,,
\end{equation}
where the upper bound in particular is the familiar {\em union bound},
or {\em Boole's inequality}. These inequalities are also best possible,
as was shown by Fr\'echet~\cite{frechet1940, frechet1943}. Boole's inequality generalizes
to the classical {\em Bonferroni inequalities}~\cite{Bonferroni1936}
\[
\begin{split}
&|\!\cup_{i\in[n]}B_i|\ \leq\!\sum_{\substack{S\subseteq [n]\\1\leq |S|\leq k}}(-1)^{|S|+1}|B_S|\quad\text{for $k\in[n]$ odd, and}\\
&|\!\cup_{i\in[n]}B_i|\ \geq\!\sum_{\substack{S\subseteq [n]\\1\leq |S|\leq k}}(-1)^{|S|+1}|B_S|\quad\text{for $k\in[n]$ even}\,,
\end{split}
\]
which however are not best possible; a substantial literature exists on 
improved and generalized Bonferroni inequalities as well as approximate
versions of the inclusion--exclusion formula (cf.~our discussion of related
work below). Our interest in this paper is the geometry and computational 
complexity of the {\em best-possible} inequalities for given 
$\emptyset\neq\mathcal{F}\subseteq 2^{[n]}\setminus\{\emptyset\}$
and $|B_S|$ for each $S\in\mathcal{F}$. In particular, we show that the 
best-possible inequalities are NP-hard to compute already 
when $\mathcal{F}$ consists of sets of size at most two.
We refer to the task of computing the best-possible inequalities as
{\em Boole's problem} in what follows. 

A geometric representation of Boole's problem
can be traced back to Hailperin~\cite{Hailperin1965}, who observed that
the best-possible inequalities can be derived by linear programming.%
\footnote{Hailperin in fact observed that linear programming can be used
to study the probability of any logical function of the events 
$B_1,B_2,\ldots,B_n$, not just the disjunction (union) of the events, 
and the given input can similarly consist of the probabilities of 
any logical functions of the events, not just conjunctions (intersections) 
of events. For simplicity, we restrict to study the complexity of 
the union probability with given intersection probabilities.}{}
Hailperin's key combinatorial observation is that it suffices to study the 
probabilities of the {\em Venn atoms} 
(the $2^n$ cells of a Venn diagram easily visualized for $n=2,3$) 
of the $n$ events $B_1,B_2,\ldots,B_n$. 
More precisely, associate with each subset $T\subseteq [n]$
the {\em atomic event} or {\em atom} 
\[
X_T=(\cap_{i\in T}B_i)\cap(\cap_{i\in[n]\setminus T}\,\Omega\!\setminus\!B_i)
\]
and observe that the atoms form a set partition of $\Omega$.
Moreover, for each $S\in\mathcal{F}$, the intersection event 
$B_S=\cap_{i\in S}B_i$ partitions into atoms as $B_S=\cup_{S\subseteq T}X_T$.
Geometrically, introducing real unknowns $x_T$ and $b_S$ for 
the pertinent probabilities, the partitioning immediately gives 
rise to the {\em Venn polytope} 
$\vp^{(\mathcal{F})}\subseteq\mathbb{Q}^{|\mathcal{F}|+2^n}$ defined 
by the hyperplanes and halfspaces
\begin{equation}
\label{eq:param-zeta}
    \sum_{T\subseteq[n]}x_T=1\,,\quad
  \sum_{S\subseteq T}x_T=b_S\ 
\text{for all $S\in\mathcal{F}$}\,,\quad\text{and}\quad
                       x_T \geq 0\ \text{for all $T\subseteq[n]$}\,.
\end{equation}
We stress that both the atom probabilities $x_T$ as well as
the intersection probabilities $b_S$ are coordinates of the Venn polytope. 

\begin{example}[A Venn polytope]
\label{exa:vp}
The Venn polytope $\vp^{(\mathcal{F})}$
for $\mathcal{F}=\{\{1\},\{2\},\{1,2\}\}$ in its 
halfspace-representation (left) and vertex-representation (right):
\[
\left[
\begin{array}{r@{\ }c@{\ }l}
x_\emptyset +
x_{\{1\}} +
x_{\{2\}} +
x_{\{1,2\}} & = & 1\\
x_{\{1\}} + x_{\{1,2\}} & = & b_{\{1\}}\\
x_{\{2\}} + x_{\{1,2\}} & = & b_{\{2\}}\\
            x_{\{1,2\}} & = & b_{\{1,2\}}\\
x_\emptyset\,,
x_{\{1\}}\,,
x_{\{2\}}\,,
x_{\{1,2\}} & \geq & 0\\
\end{array}
\right]
\qquad
\begin{array}{cccccc@{\ \ }c}
b_{\{1\}}&
b_{\{2\}}&
b_{\{1,2\}}&
x_\emptyset&
x_{\{1\}}&
x_{\{2\}}&
x_{\{1,2\}}\\\hline
0&0&0&1&0&0&0\\
1&0&0&0&1&0&0\\
0&1&0&0&0&1&0\\
1&1&1&0&0&0&1
\end{array}
\,.
\]
\end{example}
Now, following Hailperin, since the union event partitions as 
$\cup_{i\in[n]}B_i=\cup_{\emptyset\neq T\subseteq[n]}X_T=\Omega\setminus X_\emptyset$, the best-possible inequalities for the union probabability 
$u=|\!\cup_{i\in[n]}B_i|$ are
obtained from the linear program that minimizes/maximizes the objective 
function $u=\sum_{\emptyset\neq T\subseteq[n]}x_T=1-x_\emptyset$ over
the constraints \eqref{eq:param-zeta} and $b_S=|B_S|$ for 
each $S\in\mathcal{F}$ in the given input; we refer to this linear program
as {\em Hailperin's} linear program. We say that an abstract vector
$b=(b_S\in[0,1]:S\in\mathcal{F})$ of probabilities is {\em feasible}
if Hailperin's linear program is feasible. 
In particular, $b$ is feasible if and only if there exists a probability 
space $\Omega$ and $n$ events $B_1,B_2,\ldots,B_n$ in the space 
that {\em realize} $b$ with $b_S=|B_S|$ for all $S\in\mathcal{F}$.

While Hailperin's characterization of the optimal solution to Boole's problem
by linear programming is concise and elegant, 
the exponential size of the linear program in the number of events $n$ leaves 
open the question whether more efficient optimal characterizations are 
available. We show that computing the minimum/maximum union probability 
is NP-hard, thus making it unlikely that efficient optimal characterizations
are available in general. Our main results is as follows.

\begin{restatable}[Main; NP-hardness of Boole's union probability problem]{theorem}{minmaxunion}
\label{thm:minmax}
Given as input a set family
$\emptyset\neq\mathcal{F}\subseteq 2^{[n]}\setminus\{\emptyset\}$ 
and a feasible
rational vector $b=(b_S\in[0,1]:S\in\mathcal{F})$, 
each of the following two problems is NP-hard to solve:
\begin{enumerate}
\item
determine the minimum union probability for a realization of $b$, and
\item
determine the maximum union probability for a realization of $b$.
\end{enumerate}
Moreover, NP-hardness holds even when $\mathcal{F}$ consists of sets of size 
at most two.
\end{restatable}

\begin{remark}[Polynomial-time solvable special cases]
In contrast to NP-hardness, we observe that the union probability problem 
is efficiently solvable in time polynomial in the size of the input 
$(b,\mathcal{F})$, for example, when all the intersection 
probabilities are given, by the principle of inclusion and 
exclusion~\eqref{eq:inclusion-exclusion}. More generally, 
when $\mathcal{F}$ has size exponential in the number of events $n$, 
we can solve Hailperin's linear program in time polynomial in 
the assumed-exponential input size.
\end{remark}

Issues related to the complexity of Boole's problem have been 
addressed in previous research.
Pitowsky~\cite{Pitowsky1991} highlights Boole's problem and observes
its probable intractability due to its close relationship to the geometry 
of correlation polytopes.
Boros, Scozzari, Tardella, and Veneziani~\cite{BorosSTV2014}
and Boros and Lee~\cite{BorosL2026} highlight the open complexity
status of the problem even in the case of sets of size at most two.
They also observe that a closely related problem,
deciding the feasibility of a given abstract vector $b$ of
probabilities, is NP-hard, by either direct proof~\cite{Pitowsky1991} or 
reduction from the separation problem of either
a cut polytope (cf.~Deza and Laurent~\cite{DezaL1997}) or
a Boolean quadric polytope
(cf.~Jaumard, Hansen, and Poggi~de~Arag\~ao~\cite{JaumardHP1991}).
Among further closely
related NP-hard problems is the {\em probabilistic satisfiability} (PSAT)
problem of
Georgakopoulos, Kavvadias, and Papadimitriou~\cite{GeorgakopoulosKP1998},
which, given
as input $m$ clauses on $n$ Boolean variables, together with a given probability
for each of the $m$ clauses, asks whether there exists a distribution
of probabilities on the $2^n$ possible truth assignments to
the Boolean variables such that each clause is satisfied with exactly
the given probability. Moreover, Kavvadias and
Papadimitriou~\cite{KavvadiasP1990} have shown
that PSAT remains NP-hard even when all clauses consist of positive literals
only and contain at most two literals each, which also directly
establishes NP-hardness of deciding feasibility of a given abstract
vector $b$ in our present setting. Our present contribution thus is
essentially the fine-grained result that, when given a {\em feasible}
vector $b$ as input, it yet remains NP-hard to solve Hailperin's linear program
to optimality (minimum or maximum of the union probability $u=1-x_{\emptyset}$).
This contribution is motivated by data analysis applications in 
particular, where the vector $b$ is obtained from existing data, and 
hence the feasibility of $b$ is known to hold.

Geometrically, we observe that deciding whether a given abstract vector 
$b$ is feasible and deciding whether a given feasible 
vector $b$ has a realization with union probability $u$ 
both admit characterization via polytopes that are 
elementary coordinate projections of the Venn polytope $\vp^{(\mathcal{F})}$. 
We recall that coordinate projections of a polytope are easily obtained 
from the vertex representation by projecting the vertices and taking
the convex hull of the projections; we also recall Example~\ref{exa:vp}
for a Venn polytope in vertex representation.

The {\em union polytope} 
$\up^{(\mathcal{F})}\subseteq\mathbb{Q}^{|\mathcal{F}|+1}$ 
is obtained by projecting the Venn polytope $\vp^{(\mathcal{F})}$ 
to coordinates $x_\emptyset$ and $b_S$ for all $S\in\mathcal{F}$. 
The (generalized%
\footnote{When $\mathcal{F}=\binom{[n]}{2}\cup \binom{[n]}{1}$, the 
polytope $\cp^{(\mathcal{F})}\subseteq\mathbb{Q}^{|\mathcal{F}|}$ is known
as the {\em correlation polytope}~\cite{Pitowsky1991}. 
Correlation polytopes and the associated
{\em cut polytopes} are well-studied classes of polytopes related to each
other by the so-called covariance mapping; also the generalized 
correlation polytopes appear in earlier work, 
cf.~e.g.~Deza and Laurent~\cite[Chap.~I.5]{DezaL1997}.}) 
{\em correlation polytope} 
$\cp^{(\mathcal{F})}\subseteq\mathbb{Q}^{|\mathcal{F}|}$
is obtained by projecting the Venn polytope $\vp^{(\mathcal{F})}$ 
to coordinates $b_S$ for all $S\in\mathcal{F}$. 
We immediately observe the coordinate-projection chain 
$\cp^{(\mathcal{F})}\leq\up^{(\mathcal{F})}\leq\vp^{(\mathcal{F})}$.
Correlation polytopes~\cite{Pitowsky1991}, also studied more recently 
under the name {\em marginal polytopes} in the machine-learning 
community~\cite{WainwrightJ2008}, are a well-known and systematically 
studied class of polytopes via cut polytopes and cut cones in the
geometry of metrics~\cite{DezaL1997}, 
whereas the union polytopes are apparently a less studied family of 
polytopes.

\begin{example}[Union polytopes and correlation polytopes]
\label{exa:up-cp}
Halfspace-representations of union polytopes $\up^{(\mathcal{F})}$
correlation polytopes $\cp^{(\mathcal{F})}$, 
writing $u=1-x_\emptyset$ for the union probability:
\[
\begin{array}{@{}c|c@{\ }c@{\ }c@{\ }c@{}}
\mathcal{F}&
\{1,2\}&\{1\},\{2\}&\{1\},\{1,2\}&\{1\},\{2\},\{1,2\}
\\[1mm]\hline
\up^{(\mathcal{F})}&\raisebox{1mm}{$\left[\begin{array}{@{\ }r@{\ }c@{\ }l@{}}
b_{\{1,2\}}&\geq& 0\\
b_{\{1,2\}}&\leq& u\\
u&\leq& 1
\end{array}\right]$} &
\raisebox{-1mm}{$\left[\begin{array}{@{\ }r@{\ }c@{\ }l@{}}
b_{\{1\}}&\leq& u\\
b_{\{2\}}&\leq& u\\
u&\leq& b_{\{1\}}+b_{\{2\}}\\
u&\leq& 1\\
\end{array}\right]$} &
\raisebox{-1mm}{$\left[\begin{array}{@{\ }r@{\ }c@{\ }l@{}}
b_{\{1,2\}}&\geq& 0\\
b_{\{1,2\}}&\leq&b_{\{1\}}\\
b_{\{1\}}&\leq& u\\
u&\leq& 1\\
\end{array}\right]$} &
\raisebox{-3mm}{$\left[\begin{array}{@{\ }r@{\ }c@{\ }l@{}}
b_{\{1\}}+b_{\{2\}}&=&u+b_{\{1,2\}}\\
b_{\{1,2\}}&\geq& 0\\
b_{\{1,2\}}&\leq& b_{\{1\}}\\
b_{\{1,2\}}&\leq& b_{\{2\}}\\
b_{\{1\}}+b_{\{2\}}&\leq& 1+b_{\{1,2\}}\\
\end{array}\right]$}
\\[2mm]\hline
\cp^{(\mathcal{F})}&
\raisebox{4mm}{$\left[\begin{array}{c}
b_{\{1,2\}}\geq 0\\
b_{\{1,2\}}\leq 1
\end{array}\right]$} &
\left[\begin{array}{c}
b_{\{1\}}\geq 0\\
b_{\{2\}}\geq 0\\
b_{\{1\}}\leq 1\\
b_{\{2\}}\leq 1
\end{array}\right] &
\raisebox{2mm}{$\left[\begin{array}{r@{\ }c@{\ }l@{}}
b_{\{1,2\}}&\geq& 0\\
b_{\{1,2\}}&\leq& b_{\{1\}}\\
b_{\{1\}}&\leq& 1\\
\end{array}\right]$} &
\raisebox{0mm}{$\left[\begin{array}{r@{\ }c@{\ }l@{}}
b_{\{1,2\}}&\geq& 0\\
b_{\{1,2\}}&\leq& b_{\{1\}}\\
b_{\{1,2\}}&\leq& b_{\{2\}}\\
b_{\{1\}}+b_{\{2\}}&\leq& 1+b_{\{1,2\}}\\
\end{array}\right]$}
\end{array}\,.
\]
In particular, we observe both the Boole--Frech\'et 
bound \eqref{eq:boole-frechet} in $\up^{(\{1\},\{2\})}$ and 
the inclusion--exclusion formula~\eqref{eq:inclusion-exclusion} for $n=2$ 
in $\up^{(\{1\},\{2\},\{1,2\})}$ above. We postpone a short generic 
analysis of the geometry of the polytopes to Section~\ref{sect:preliminaries}.
\end{example}

Geometrically, an equivalent formulation of Boole's problem is 
now easily seen to be the following: given a point $b$ 
in a correlation polytope $\cp^{(\mathcal{F})}$ as input, find the preimage 
of $b$ under coordinate projection in the corresponding union 
polytope $\up^{(\mathcal{F})}$. This preimage is exactly 
the interval of the union probability $u=1-x_{\emptyset}$ and is particularly 
simple to find by substituting $b$ into the halfspaces 
of $\up^{(\mathcal{F})}$; the reader is encouraged to try this out with 
the union polytopes $\up^{(\mathcal{F})}$ in Example~\ref{exa:up-cp} above.

In spite of the appealing simplicity of the small examples above, 
it is however unlikely that the correlation and union polytopes in general 
admit a simple description for their halfspace-representations. 
Indeed, as discussed above, deciding the feasibility, or, what is the same,
deciding membership in the correlation polytope, for a given abstract
vector $b$ is known to be NP-hard. For completeness, we state this fact as
a theorem. 

\begin{restatable}[Hardness of membership in a correlation polytope;~\cite{DezaL1997,JaumardHP1991,KavvadiasP1990,Pitowsky1991}]{theorem}{feasibility}
\label{thm:cp}
Given a set family
$\emptyset\neq\mathcal{F}\subseteq 2^{[n]}\setminus\{\emptyset\}$ and 
a rational vector
$b=(b_S\in[0,1]:S\in\mathcal{F})$ as input, it is NP-hard to decide 
whether $b\in\cp^{(\mathcal{F})}$. Moreover, this holds even when 
$\mathcal{F}$ consists of sets of size at most two. 
\end{restatable}

Similarly, deciding membership in a union polytope is NP-hard; we prove the
following theorem in the appendix as a corollary of the proof of 
Theorem~\ref{thm:minmax}. 

\begin{restatable}[Hardness of membership in a union polytope]{theorem}{unionmember}
\label{thm:up}
Given a set family 
$\emptyset\neq\mathcal{F}\subseteq 2^{[n]}\setminus\{\emptyset\}$, 
a rational vector $b=(b_S\in[0,1]:S\in\mathcal{F})$, and a rational
number $x_\emptyset\in[0,1]$ as input, it is NP-hard to 
determine whether $(b,x_\emptyset)\in\up^{(\mathcal{F})}$.
Moreover, this holds even when $\mathcal{F}$ consists of sets of size 
at most two and $b\in\cp^{(\mathcal{F})}$.
\end{restatable}

In contrast to the complexity of the halfspace-representations of the 
union and correlation polytopes, the vertex-representations of the polytopes
admit description by coordinate projection from the vertices of the Venn 
polytopes; we record, and prove in the appendix, basic geometric 
facts about the polytopes in 
Propositions~\ref{prop:venn}~to~\ref{prop:union} in 
Section~\ref{sect:preliminaries} as preliminaries to the
proof of our main theorem in Section~\ref{sect:hardness}. 
One immediate consequence of the zero-one vertex-representations observed in 
Propositions~\ref{prop:venn}~to~\ref{prop:union} 
is that the membership problems in Theorems~\ref{thm:cp}~and~\ref{thm:up}
are in fact NP-complete by Carath\'eodory's theorem.

Let us now give a short overview of the proof of our main theorem
(Theorem~\ref{thm:minmax}). For NP-hardness of the minimum union probability
problem, we proceed with a straightforward reduction from the fractional graph
coloring problem~\cite{GrotschelLS1981,LundY1994,ScheinermanU1997}. 
The NP-hardness
of the maximum union probability problem requires somewhat more work. 
Our approach is to transform a restricted linear program for maximum
union probability in steps to a variant of its dual linear program whose 
polyhedron $\phi_G$ has facets that enable eventual relation to 
the $k$-clique problem on a given graph $G$. We establish this relation 
with standard techniques for rational polyhedra~\cite{GrotschelLS1993} by 
transforming a validity oracle (obtainable from an oracle for maximum 
union probability) for $\phi_G$ via the polar polytope $\phi_G^*$ and its polar 
$\phi_G^{**}=\phi_G$ to a membership oracle for $\phi_G$, which enables
immediate solution of the $k$-clique problem on $G$.

We refer to the monographs of Gr\"otschel, Lov\'asz, and 
Schrijver~\cite{GrotschelLS1993} for pertinent terminology on rational
polyhedra and Ziegler~\cite{Ziegler2000} for zero-one polytopes; see
also Schrijver~\cite{Schrijver2003A,Schrijver2003B,Schrijver2003C}.
For fractional graph theory, we refer to Scheinerman and 
Ullman~\cite{ScheinermanU1997}.

\subsection*{Related work}

\paragraph*{Generalized Bonferroni inequalities and local lemmata in probability and logic.}

A large body of work studies generalized Bonferroni~\cite{Bonferroni1936} inequalities, with or without Hailperin's~\cite{Hailperin1965} linear programming formulation. Determining the best Bonferroni bounds for binomial moments $S_k=\sum_{K\in\binom{[n]}{k}}|B_K|$ given up to order $k \leq 2$ is studied by~Chung and Erdős~\cite{chungErdos1952}, Dawson and Sankoff~\cite{dawsonSankoff1967inequality}, Galambos~\cite{Galambos1977}, Kounias and Marin~\cite{ Kounias1968, KouniasMarin1976},  Kwerel~\cite{Kwerel1975}, Sathe et al. \cite{SathepPradhanShah1980}, Hunter~\cite{Hunter1976}, Worsley~\cite{Worsley1982}, de Caen~\cite{Caen1997lower}, Kuai et al. \cite{kuai2000lower, Kuai2000}, Frolov~\cite{frolov2012bounds},  and Yang et al.~\cite{YangAlajajiT2014}. More generalized versions for $k \geq 2$ were also studied by Gallot~\cite{gallot1966bound}, Prékopa and Gao~\cite{prekopa2005bounding}, Galambos and Mucci~\cite{galambos1Mucci980}, and Kwerel~\cite{Kwerel1975} providing several sharp inequalities. A generalized version of the union problem was studied by Galambos~\cite{Galambos1977}, Platz~\cite{platz1985sharp}, Prékopa~\cite{prekopa1988, prekopa1992inequalities}, Boros and Prékopa~\cite{BorosPrekopa1989}, Galambos and Simonelli~\cite{GalambosS1996}, Min-Young~\cite{MinYoung1991}, Sibya~\cite{Sibya1992}, Spencer~\cite{Spencer1994}, and Dohmen~\cite{Dohmen2003} for $\Pr(X \geq t)$; where $X$ is a random variable that counts the number of events occur among $n$ many events and $t \in [n]$. Hailperin's seminal study~\cite{Hailperin1965} was followed by a more systematic study of sentential probability logic~\cite{Hailperin1996}; Nilsson~\cite{nilsson1986} models probabilities of logical sentences, where union probabilities arise from combining dependent events under consistency constraints. Jaeger’s~\cite{jaeger2001} probabilistic inference rules automates such tight bounds on the target probability as a function of the input constraints.
In a related setting, bounding the probability of the union of $n$ bad events, the Lov\'{a}sz Local Lemma~\cite{lovaszErdos1973} shows that if each bad event has probability at most $p$ and depends on at most $d$ others, then the union probability is strictly less than $1$, hence a good outcome exists, provided that $p(d+1) \le \exp(-1)$; a constructive randomized version was later given by Moser and Tardos~\cite{tardosmoser2010}. Barvinok \cite{Barvinok2025} recently showed that, under a stricter condition, intersections of order $O(\log(n/\epsilon))$ suffice to approximate the good event's probability up to relative error $\epsilon$.

\paragraph*{Approximate and exact inclusion--exclusion in algorithms and complexity.}

Linial and Nisan~\cite{LinialN1990} present an {\em approximate} level-wise variant of the inclusion--exclusion formula \eqref{eq:inclusion-exclusion} relying on Chebyshev polynomials and use it to study distributions that fool constant-depth circuits, followed by work of~Luby and Veli\v{c}kovi\'c~\cite{luby1991deterministic}, Bazzi~\cite{Bazzi2009}, Razborov \cite{Razborov2009} and Braverman~\cite{Braverman2010}; a related problem is to approximate the number of satisfying assignments for a DNF with $m$ clauses and $n$ variables was studied by Luby and Veli\v{c}kovi~\cite{luby1991deterministic}, Kahn et al.~\cite{KahnLS1996} show logarithmic-size intersections suffice for uniqueness of the DNF assignment, and later Melkman and Shimony~\cite{MelkmanS1997} study its construction. Further works that improve the Linial-Nisan framework in specific applications include Klein and Metsch~\cite{KleinM2007} (cryptography) and Shrestov~\cite{Sherstov2008} (boolean symmetric functions).
The {\em exact} principle of inclusion and exclusion and its generalization
to M\"obius inversion on partially ordered 
sets~\cite{Stanley2012} other than $(2^{[n]},\subseteq)$ 
are fundamental primitives in the study fine-grained and parameterized 
algorithms, cf.~Fomin and Kratsch~\cite{FominK2010}
and Cygan et al.~\cite{CyganFKLMMPS2015}, 
including the study of canonical hard problems such as matrix
permanent~\cite{Ryser1963} and graph coloring~\cite{BjorklundHK2009}.

\paragraph*{Data mining.}

Downwards-closed set families arise naturally in data mining, in the task of finding association rules or frequent sets from zero-one-valued data. 
In that problem, the data is a (large) zero-one matrix $D$, and the task is to find all sets $S$ of columns such that at least a fraction $\epsilon$ of the rows of the matrix have a $1$-element in each column of $S$. In our notation, the task is to find the family $\mathcal{F} = \{ S \subseteq [n] : b_S \geq \epsilon \}$. In applications, the matrix $D$ is typically sparse, and the zeros and ones have an asymmetric role; hence, for suitably chosen~$\epsilon$, the family $\mathcal{F}$ will have a reasonable size. The algorithms for finding the collection $\mathcal{F}$ are based on a level-wise search: first find the singleton sets satisfying the criterion, and when the sets of size $k$ are known, build all candidate sets $C$ of size $k+1$ such that the subsets of $C$ of size $k$ satisfy the criterion, and check whether $C$ satisfies the criterion. This so-called {\em Apriori} algorithm works well in many cases. The early papers and algorithms for this include Agrawal et al.~\cite{AgrawalS1993}, Agrawal and Srikant~\cite{AgrawalS1994}, Mannila et al.~\cite{MannilaT1994}, Agrawal et al.~\cite{AgrawalM1996}, and have since been followed by many others.\footnote{On March 26, 2026, a search on Scopus for ``association rule mining'' returned 20,581 documents.} Theoretical analyses of the question also exist, but to our knowledge, the question we study in this paper has not been considered in the frequent set literature. Note that in the setting of frequent sets, the vector $b$ corresponding to the collection $\mathcal{F}$ is known to be feasible, so Boole’s problem is significant exactly in the case where the feasibility of $b$ is known.

\paragraph*{Further polytopes and extension complexity.}

Wainwright and Jordan~\cite{WainwrightJ2008} derive the notion of a marginal polytope by starting from a set of sufficient statistics and the corresponding mean parameters (the expectations of the sufficient statistics). Then the set of all distributions that realize the mean parameters is convex and hence can be described by inequalities, giving the marginal polytope. It is mostly used in the case where the sufficient statistics correspond to singletons or pairs of variables in the underlying probability space, but one could also consider, e.g., the case where there is a sufficient statistic for each intersection probability. Then the marginal polytope for this set of sufficient statistics is the correlation polytope. Related to correlation and union polytopes, Yang, Alajaji, and Takahara~\cite{YangAT2015} seek to bound the union probability from weighted-aggregate input of the form $\{|B_i|\}_{i=1}^n$ and $\left\{\sum_{j=1}^n c_j\,|B_i\cap B_j|\right\}_{i=1}^n$, which geometrically corresponds to intersecting the union polytope with the affine subspace defined by the input and studying the resulting polytope; whether our present NP-hardness results extend to exact optimization in such intersection polytopes remains an interesting direction for future work. A yet further line of work is the extension complexity of the correlation polytopes, cf.~Aboulker et al.~\cite{AboulkerFHMS2019}.

\section{Geometry of the Venn, union, and correlation polytopes}

\label{sect:preliminaries}

The following propositions collect some immediate observations 
on the geometry of the Venn, union, and correlation polytopes. 
The proofs of the propositions appear in the appendix.
We need short preliminaries.
For a nonnegative integer $n$, 
we write $[n]=\{1,2,\ldots,n\}$ and~$2^{[n]}$ 
for the set of all subsets of $[n]$.
We write $\mathbb{Q}$ for the set of rational numbers. 
For a logical proposition $P$, let us work with 
{\em Iverson's bracket notation} and set $\iv{P}=1$ if $P$ is true 
and $\iv{P}=0$ if $P$ is false. 
Let $\emptyset\neq\mathcal{F}\subseteq 2^{[n]}$. 
Define the set-inclusion indicator 
vectors $\zeta^{(\mathcal{F},T)}\in\mathbb{Q}^{|\mathcal{F}|}$ for all 
$S\in\mathcal{F}$ and $T\subseteq [n]$ by the rule 
\begin{equation}
\label{eq:set-inclusion-indicator}
\zeta^{(\mathcal{F},T)}_S=\iv{S\subseteq T}\,.
\end{equation}
Let us write $e^{(T)}\in\mathbb{Q}^{2^n}$ for the standard basis vector defined 
for all $T,U\subseteq[n]$ by the rule $e^{(T)}_U=\iv{T=U}$.
Let us write 
$\uparrow\!S=\{T:S\subseteq T\subseteq [n]\}$ for the {\em up-set} of a set
$S\in\mathcal{F}$. For a set $V\subseteq\mathbb{Q}^d$, we write $\conv(V)$ for
the convex hull of $V$.

\begin{restatable}[Geometry of Venn polytopes]{proposition}{vpgeom}
\label{prop:venn}
We have the following.
\begin{enumerate}
\item
The Venn polytope 
$\vp^{(\mathcal{F})}\subseteq\mathbb{Q}^{|\mathcal{F}|+2^n}$
has dimension $2^n-1$ and has only zero-one-valued vertices
which are given by the representation
$\vp^{(\mathcal{F})}=\conv((\zeta^{(\mathcal{F},T)},e^{(T)}):T\subseteq[n])$.
\item
The number of vertices of $\vp^{(\mathcal{F})}$ is $2^n$ and 
a set of\/ $1+|\mathcal{F}|$ hyperplanes for $\vp^{(\mathcal{F})}$
is given by the equalities in \eqref{eq:param-zeta}.
\end{enumerate}
\end{restatable}

\begin{restatable}[Geometry of correlation polytopes]{proposition}{cpgeom}
\label{prop:correlation}
We have the following.
\begin{enumerate}
\item
The correlation polytope 
$\cp^{(\mathcal{F})}\subseteq\mathbb{Q}^{|\mathcal{F}|}$
is full-dimensional and has only zero-one-valued vertices which are 
given by the representation 
$\cp^{(\mathcal{F})}=\conv(\zeta^{(\mathcal{F},T)}:T\subseteq [n])$.
\item
The number of vertices of $\cp^{(\mathcal{F})}$ is 
$|\!\bigwedge_{S\in\mathcal{F}}\,\{\uparrow\!S,2^{[n]}\setminus\!\uparrow\!S\}|$,
where the meet is taken in the set partition lattice of the set\/ $2^{[n]}$.
In particular, the number of vertices is $2^n$ if and only if 
$\{k\}\in\mathcal{F}$ for all $k\in[n]$. 
\item
A vector $b=(b_S\in[0,1]:S\in\mathcal{F})$ is feasible
if and only if $b\in\cp^{(\mathcal{F})}$. 
\end{enumerate}
\end{restatable}

\begin{restatable}[Geometry of union polytopes]{proposition}{upgeom}
\label{prop:union}
We have the following.
\begin{enumerate}
\item
The union polytope 
$\up^{(\mathcal{F})}\subseteq\mathbb{Q}^{|\mathcal{F}|+1}$ 
has dimension 
$|\mathcal{F}|+\iv{\mathcal{F}\neq 2^{[n]}\setminus\{\emptyset\}}$
and has only zero-one-valued vertices which are given by the 
representation 
$\up^{(\mathcal{F})}=\conv((\zeta^{(\mathcal{F},T)},\iv{T=\emptyset}):T\subseteq[n])$.
\item
When $\up^{(\mathcal{F})}$ has codimension $1$ 
it has the hyperplane
$x_\emptyset+\sum_{\emptyset\neq S\subseteq[n]}(-1)^{|S|+1}b_S = 1$.
\item
The number of vertices of $\up^{(\mathcal{F})}$ either
agrees with the number of vertices of $\cp^{(\mathcal{F})}$
or is one more than this; the latter case occurs if and only if 
there exists an $k\in[n]$ with $\{k\}\notin\mathcal{F}$.
\item
For every $b=(b_S\in[0,1]:S\in\mathcal{F})$ and $x_\emptyset\in[0,1]$
we have that $(b,x_\emptyset)\in\up^{(\mathcal{F})}$ if and only if 
$b$~is realizable and there exists a realization $x$ of $b$ with union 
probability $u=1-x_\emptyset$. 
\end{enumerate}
\end{restatable}

\section{Hardness of optimally solving Boole's problem}

\label{sect:hardness}

This section restates and proves our main theorem. 

\minmaxunion*

For convenience,
we present the proof in two parts, one part for the minimum union
probability, and one part for the maximum union probability. Before proceeding
with the first part, let us recall the fractional chromatic number problem
(e.g.~\cite[Chap.~3]{ScheinermanU1997}), which is known to 
be NP-hard~\cite{LundY1994}. Let us write $V(G)$ for the vertex
set and $\mathcal{I}(G)$ for the 
set of all nonempty independent sets of a graph $G$.
Given a graph $G$ as input, we are to 
compute the minimum value of $\sum_{I\in\mathcal{I}(G)}y_I$
subject to the constraints 
$\sum_{u\in I\in\mathcal{I}(G)}y_I\geq 1$ for all $u\in V(G)$
and $y_I\geq 0$ for all $I\in\mathcal{I}(G)$. This minimum value 
is the {\em fractional chromatic number} $\chi_f(G)$ of $G$. 
A map $y:\mathcal{I}(G)\rightarrow\mathbb{Q}$ that satisfies the constraints
is a {\em fractional coloring} of $G$.

\begin{proof}[Proof of Theorem~\ref{thm:minmax}, Item 1 (minimum union probability)]
We reduce from the fractional chromatic number problem. 
Let $G$ be a graph given as input.
Without loss of generality we may assume that the vertex set $V(G)=[n]$
and that the edge set $E(G)$ consists of 2-element subsets of $V(G)$.
Set 
\[
\mathcal{F}=\{\{u\}:u\in V(G)\}\cup E(G)\subseteq 2^{[n]}\,.
\]
Define $b:\mathcal{F}\rightarrow\mathbb{Q}$ by setting
\begin{alignat}{2}
\label{eq:vertices-one-over-n}
b_{\{u\}}&=\frac{1}{n}&\quad&\text{ for all $u\in V(G)$, and}\\
\label{eq:edges-zero}
b_{\{u,v\}}&=0&\quad&\text{ for all $\{u,v\}\in E(G)$}.
\end{alignat}
We immediately observe that $(b,\mathcal{F})$ is feasible; indeed,
to obtain a realization $x$, for each $T\subseteq V(G)=[n]$ we can set
\begin{equation}
\label{eq:min-union-feasible}
x_T=\begin{cases}
\frac{1}{n} &\text{if $|T|=1$},\\
0 & \text{otherwise}.
\end{cases}
\end{equation}
Suppose now that $x^*=(x_T^*\in [0,1]:T\subseteq[n])$ is a realization of $b$ 
that minimizes the union probability 
$\sum_{\emptyset\neq T\subseteq [n]}x_T^*$. Define $y=(y_T\in[0,1]:\emptyset\neq T\subseteq V(G))$ for all $\emptyset\neq T\subseteq V(G)$ by setting
$y_T=n\cdot x^*_T$. We claim that $y$ is zero outside of $\mathcal{I}(G)$
and a fractional coloring of $G$ when restricted to $\mathcal{I}(G)$;
thus $\sum_{\emptyset\neq T\subseteq V(G)}y_T\geq\chi_f(G)$ by definition
of fractional coloring. 
Indeed, from~\eqref{eq:edges-zero} and $x^*$ realizing $b$ 
we observe that $y_T=0$ unless $T\in\mathcal{I}(G)$. 
Moreover, from~\eqref{eq:vertices-one-over-n} and $x^*$ realizing $b$ 
we observe that $\sum_{u\in T\subseteq V(G)}y_T=1$ holds for all $u\in V(G)$.
Thus, the minimum union probability of $(b,\mathcal{F})$ is at least 
$\chi_f(G)/n$. Conversely, suppose that
$y^*$ is a fractional coloring of $G$ with 
$\sum_{I\in\mathcal{I}(G)}y_I=\chi_f(G)$. Extend this fractional coloring
to all subsets $\emptyset\neq T\subseteq[n]$ by setting $y^*_T=0$ 
if $T\notin\mathcal{I}(G)$.
By definition of fractional coloring, we have that  
$\sum_{u\in T\subseteq V(G)}y_T^*\geq 1$. Since all nonempty subsets 
of independent sets in $G$ are independent, without loss of generality we 
can assume that $\sum_{u\in T\subseteq V(G)}y_T^*=1$ holds for all $u\in V(G)$;
indeed, to ensure this we can iterate over all $u_0\in V(G)$ and do 
the following: 
while the slack $\left(\sum_{u_0\in T\in\mathcal{I}(G)} y_T^*\right)-1$ is 
positive, greedily transfer mass from $y_T^*$ to $y_{T\setminus\{u_0\}}^*$ 
for sets $u_0\in T\in\mathcal{I}(G)$ with $|T|\ge 2$ and greedily reduce mass 
in $y_{\{u_0\}}^*$, until the slack is zero. It is immediate that the greedy 
process zeroes slack at $u_0$ only and leaves other $u\neq u_0$ 
untouched in terms of their slack; furthermore, 
$\sum_{T\in\mathcal{I}(G)} y_T^*$ does not increase in the process.

Now define $x=(x_T\in[0,1]:T\subseteq [n])$ by setting 
$x_\emptyset=1-\frac{1}{n}\sum_{\emptyset\neq T\subseteq[n]}y_T^*$
and $x_T=\frac{1}{n}\cdot y_T^*$ for all $\emptyset\neq T\subseteq[n]$;
since fractional chromatic number always satisfies $\chi_f(G)\leq n$,
in particular we conclude that $x_\emptyset\geq 0$. 
From \eqref{eq:vertices-one-over-n} and 
\eqref{eq:edges-zero}
we observe that $x$ realizes $(b,\mathcal{F})$ with union probability
that equals $\chi_f(G)/n$.
We conclude that the minimum union probability of $(b,\mathcal{F})$ 
equals $\chi_f(G)/n$.
\end{proof}

Let us now turn to the proof of Theorem~\ref{thm:minmax}(2), where
the NP-hardness reduction will eventually be from the $k$-clique problem. 
We start with a lemma that shows an oracle for the linear program for 
the maximum union probabability problem on given input $(b,\mathcal{F})$ 
can be used to solve a graph-based linear program, which the subsequent 
proof will then use to solve $k$-clique.

\begin{lemma}[A transformed special case of the dual linear program]
\label{lem:special-dual}
Assume access to an oracle that solves the maximum union
probability problem for a set family $\emptyset\neq\mathcal{F}\subseteq2^{[n]}$ 
and a feasible rational $b:\mathcal{F}\rightarrow [0,1]$ given as query.
Then, given a nonempty $n$-vertex graph $G$ and edge weights
$w:E(G)\rightarrow\mathbb{Q}_{\geq 0}$
as input, we can solve the linear program to
\begin{equation}
\label{eq:dual-simplified}
\begin{split}
&\text{maximize $\sum_{\{u,v\}\in E(G)}w_{\{u,v\}}y_{\{u,v\}}$}\\
&\text{subject to $\sum_{E(G)\ni \{u,v\}\subseteq T}y_{\{u,v\}}\leq |T|-1$ 
for all cliques $\emptyset\neq T\subseteq V(G)$}
\end{split}
\end{equation}
for its optimum value in time polynomial in $n$ and the rational encoding 
length of $w$, by making a single query to the maximum union probability 
oracle, with $\mathcal{F}$ consisting of sets of size at most two.
\end{lemma}

\begin{proof}
Let us recall the maximum union probability problem as a linear program 
on a given query $(b,\mathcal{F})$ with $b$ feasible. Namely, we are to
\begin{equation}
\label{eq:max-union-lp}
\begin{split}
&\text{maximize $\sum_{\emptyset\neq T\subseteq[n]}x_T$}\\
&\text{subject to $\sum_{T\subseteq[n]}x_T=1$, $\sum_{S\subseteq T}x_T=b_S$ for all $S\in\mathcal{F}$, and $x_T\geq 0$ for all $T\subseteq[n]$}\,. 
\end{split}
\end{equation}
We assume we have an oracle subroutine that on query
$(b,\mathcal{F})$ with $b$ feasible 
outputs the value of \eqref{eq:max-union-lp}.

Let $G$ be a graph and $w:E(G)\rightarrow\mathbb{Q}_{\geq 0}$ a function
given as input. (We will discuss the input $w$ in detail only towards
the end of the proof of this lemma.) 
We may assume that $n\geq 2$ and that $G$ has at least one edge.
We may also assume that the vertex set $V(G)=[n]$ and that the edge
set $E(G)$ consists of $2$-element subsets of $V(G)$.
Fix
\[
\mathcal{F}=\{\{u\}:u\in V(G)\}\cup\{\{u,v\}:u,v\in V(G)\}\,.
\]
We will also fix $b_{\{u\}}=\frac{1}{n}$ and $b_{\{u,v\}}=0$ 
for all distinct $u,v\in V(G)$ with $\{u,v\}\notin E(G)$ in our
query to~\eqref{eq:max-union-lp} in what follows. 
With this fixing, the query $(b,\mathcal{F})$ 
to~\eqref{eq:max-union-lp} that we will make in what follows 
is defined by a function $c:E(G)\rightarrow[0,\frac{1}{n^2}]$ when we set
$b_{\{u,v\}}=c_{\{u,v\}}$ for all $\{u,v\}\in E(G)$; the upper bound 
$\frac{1}{n^2}$ for the values of $c$ is to ensure feasibility of the
resulting query $b$. Indeed, to see that such a function $c$ results 
in a feasible query $b$ to~\eqref{eq:max-union-lp}, observe that $b$ 
is realized by the atom probabilities $x$ defined for all $T\subseteq[n]$ by
\[
x_T=\begin{cases}
0 & \text{if $|T|>2$, or $|T|=2$ and $T\notin E(G)$},\\
c_T & \text{if $|T|=2$ and $T\in E(G)$},\\
\frac{1}{n}-\sum_{T\subsetneq U\subseteq[n]}x_U & \text{if $|T|=1$},\\
1-\sum_{\emptyset\neq U\subseteq[n]}x_U & \text{if $T=\emptyset$}.
\end{cases}
\]
Now let us assume a function $c:E(G)\rightarrow[0,\frac{1}{n^2}]$ and
thus a query $(b,\mathcal{F})$ has been fixed. We will proceed along a
sequence of linear programs equivalent to \eqref{eq:max-union-lp}. 
Since we have 
fixed $b_{\{u\}}=\frac{1}{n}$ for all $u\in V(G)$ and $x_T\geq 0$ holds 
for all $\emptyset\neq T\subseteq[n]$, we observe that 
\[
\sum_{\emptyset\neq T\subseteq[n]}x_T\leq \sum_{u\in V(G)}\sum_{u\in T\subseteq V(G)}x_T\leq n\cdot\frac{1}{n}=1\,. 
\]
This inequality in particular implies that we can remove the constraint
$\sum_{T\subseteq[n]}x_T=1$ and the variable $x_\emptyset$ 
from \eqref{eq:max-union-lp}; that is, an equivalent linear program is to
\begin{equation}
\label{eq:max-union-lp-trim}
\begin{split}
&\text{maximize $\sum_{\emptyset\neq T\subseteq[n]}x_T$}\\
&\text{subject to $\sum_{S\subseteq T}x_T=b_S$ for all $S\in\mathcal{F}$
and $x_T\geq 0$ for all $\emptyset\neq T\subseteq[n]$}.
\end{split}
\end{equation}
The dual of \eqref{eq:max-union-lp-trim} is to
\begin{equation}
\label{eq:max-union-lp-trim-dual}
\begin{split}
&\text{minimize $\sum_{S\in\mathcal{F}}b_Sy_S$}\\
&\text{subject to $\sum_{\mathcal{F}\ni S\subseteq T}y_S\geq 1$ for all $\emptyset\neq T\subseteq[n]$},
\end{split}
\end{equation}
where each variable $y_S$ for $S\in\mathcal{F}$ can take arbitrary rational
values. Since $b_{\{u,v\}}=0$ whenever $\{u,v\}\notin E(G)$, we observe
that in \eqref{eq:max-union-lp-trim-dual} each constraint for 
$\emptyset\neq T\subseteq[n]$ that is not a clique in $G$ can be 
trivially satisfied by assigning an arbitrarily large value to a 
variable $y_{\{u,v\}}$ for $T\supseteq\{u,v\}\notin E(G)$, 
without affecting the value of the objective function since $b_{\{u,v\}}=0$.
In particular, we can remove all the non-clique constraints and all the
variables $y_{\{u,v\}}$ with $\{u,v\}\notin E(G)$ 
from \eqref{eq:max-union-lp-trim-dual} to obtain the equivalent linear 
program to 
\begin{equation}
\label{eq:max-union-lp-trim-dual-clique}
\begin{split}
&\text{minimize $\frac{1}{n}\sum_{u\in V(G)}y_{\{u\}}+\sum_{\{u,v\}\in E(G)}c_{\{u,v\}}y_{\{u,v\}}$}\\
&\text{subject to $\sum_{V(G)\ni u\in T}y_{\{u\}}+\sum_{E(G)\ni \{u,v\}\subseteq T}y_{\{u,v\}}\geq 1$ for all cliques $\emptyset\neq T\subseteq V(G)$}.
\end{split}
\end{equation}
Now fix an arbitrary optimum solution $y^*$ to \eqref{eq:max-union-lp-trim-dual-clique}. Without loss of generality we can assume that 
$y^*_{\{u\}}=1$ holds for all $u\in V(G)$; 
indeed, $y^*_{\{u\}}\geq 1$ is immediate by the clique $T=\{u\}$, and 
the assumption $y^*_{\{u_0\}}=1+\epsilon$ for some $u_0\in V(G)$
and $\epsilon>0$ yields a contradiction to the optimality of 
$y^*$ by decreasing $y^*_{\{u_0\}}$ (with objective coefficient $\frac{1}{n}$) 
by $\epsilon$ and increasing each $y^*_{\{u_0,v\}}$ 
(with objective coefficient $c_{\{u_0,v\}}\leq\frac{1}{n^2}$)
by $\epsilon$ for each edge $\{u_0,v\}\in E(G)$ 
to obtain an overall decrease in the minimization objective of 
\eqref{eq:max-union-lp-trim-dual-clique} without 
violating any of the clique constraints.
Thus, a linear program equivalent to \eqref{eq:max-union-lp-trim-dual-clique}
is to
\begin{equation}
\label{eq:dual-simplified-1}
\begin{split}
&\text{minimize $1+\sum_{\{u,v\}\in E(G)}c_{\{u,v\}}y_{\{u,v\}}$}\\
&\text{subject to $|T|+\sum_{E(G)\ni \{u,v\}\subseteq T}y_{\{u,v\}}\geq 1$ 
for all cliques $\emptyset\neq T\subseteq V(G)$}. 
\end{split}
\end{equation}
Recalling that the variables 
$y_{\{u,v\}}$ can take arbitrary rational values, changing the sign 
of the variables and removing the known offset $1$ from the objective, 
we observe that an offset-$1$-equivalent linear program 
to \eqref{eq:dual-simplified-1} is to 
\begin{equation}
\label{eq:dual-simplified-2}
\begin{split}
&\text{maximize $\sum_{\{u,v\}\in E(G)}c_{\{u,v\}}y_{\{u,v\}}$}\\
&\text{subject to $\sum_{E(G)\ni \{u,v\}\subseteq T}y_{\{u,v\}}\leq |T|-1$ 
for all cliques $\emptyset\neq T\subseteq V(G)$}.
\end{split}
\end{equation}
It follows immediately from this sequence \eqref{eq:max-union-lp}--\eqref{eq:dual-simplified-2} of offset-equivalent linear programs
that we can use our assumed maximum union probability oracle to obtain 
the maximum value of~\eqref{eq:dual-simplified} 
for an arbitrary $c:E(G)\rightarrow [0,\frac{1}{n^2}]$.

Now consider the given rational input $w:E(G)\rightarrow \mathbb{Q}_{\geq 0}$
and recall that we are to solve~\eqref{eq:dual-simplified} for its
optimum value. 
When $w$ is identically zero, the optimum value of~\eqref{eq:dual-simplified} 
is zero, so we can assume that $\max_{e\in E(G)}w(e)>0$ in what follows. 
Set $\alpha=(n^2\max_{e\in E(G)}w_e)^{-1}$ and $c=\alpha w$, observing that
all values of $c$ are rational in the interval $[0,\frac{1}{n^2}]$. 
Furthermore, $c$ has rational encoding length at most 
polynomial in $n$ and the rational encoding length of~$w$. 
Now invoke the maximum union probability oracle on input $(b,\mathcal{F})$
derived from $c$ and offset the oracle output by~$1$ to obtain the optimum 
value $t$ for~\eqref{eq:dual-simplified-2}, and return the value 
$t/\alpha$, observing that the returned value is the optimum value
of~\eqref{eq:dual-simplified} on input $w$. This process runs in time
bounded by a polynomial in $n$ and the rational encoding length of $w$,
making a single query $(b,\mathcal{F})$ to the maximum union probability 
oracle with $\mathcal{F}$ consisting of sets of size at most two. 
\end{proof}
We are now ready for the main reduction from the $k$-clique problem.

\begin{proof}[Proof of Theorem~\ref{thm:minmax}, Item 2 (maximum union probability)]

We reduce from the $k$-clique problem. Let $G$ be a graph and $k$ be 
a nonnegative integer given as input. Without loss of generality we may
assume that $n\geq k\geq 2$ and that $G$ has at least one edge.
We may also assume that the vertex set $V(G)=[n]$ and that the edge
set $E(G)$ consists of $2$-element subsets of $V(G)$. 
It now follows immediately from Lemma~\ref{lem:special-dual} that
we have a subroutine to find the maximum value of any
objective $w:E(G)\rightarrow \mathbb{Q}_{\geq 0}$ over the polyhedron 
\begin{equation}
\label{eq:phi-polyhedron}
\phi_G=\biggl\{y:E(G)\rightarrow\mathbb{Q}\ :\ 
\text{$\sum_{E(G)\ni \{u,v\}\subseteq T}y_{\{u,v\}}\leq |T|-1$ 
for all  cliques $\emptyset \neq T\subseteq V(G)$}\biggr\}\,
\end{equation}
We observe in particular that $\phi_G$ is unbounded, in fact with cone
$\mathbb{Q}^{|E(G)|}_{\leq 0}$ as we will see, so it will be convenient
to proceed via the polar to arrive at the $k$-clique problem.

We first construct a {\em validity oracle} oracle for 
$\phi_G$; 
namely, given $\mu:E(G)\rightarrow\mathbb{Q}$ and $\nu\in\mathbb{Q}$
as input, we are to assert whether $\mu^\top y\leq \nu$ holds 
for all $y\in\phi_G$. On given input $\mu,\nu$, the validity oracle proceeds
as follows. Unless $\mu$ takes only nonnegative values, assert that
$\mu^\top y\leq \nu$ does not hold for all $y\in\phi_G$; indeed,
if $\mu(\{u,v\})<0$ for some $\{u,v\}\in E(G)$, construct 
a vector $y_0:E(G)\rightarrow\mathbb{Q}$ that 
satisfies $y_0(\{u,v\})<\min(0,\nu/\mu(\{u,v\}))$ and is zero elsewhere 
to witness $\mu^\top y_0>\nu$ as well as $y_0\in\phi_G$.
When $\mu$ takes only nonnegative values, use the subroutine with 
$w=\mu$ to obtain the maximum value $t$ of the objective $w^\top y$ 
for $y\in\phi_G$, and assert that $\mu^\top y\leq \nu$ holds 
for all $y\in\phi_G$ if $t\leq \nu$; otherwise assert that 
$\mu^\top y\leq \nu$ does not hold for all $y\in\phi_G$.
This completes the description of the validity oracle for $\phi_G$.

The validity oracle for $\phi_G$ immediately gives a {\em membership oracle} 
for the {\em polar}
\begin{equation}
\label{eq:polar}
\phi_G^*=\biggl\{z:E(G)\rightarrow\mathbb{Q}\ :\ \text{$z^\top y\leq 1$ for all $y\in\phi_G$}\biggr\}
\end{equation}
of $\phi_G$; indeed, by definition of the polar, 
for $\mu\in E(G)\rightarrow\mathbb{Q}$ given as input we have 
$\mu\in\phi_G^*$ if and only if $\mu^\top y\leq 1$ holds for all $y\in\phi_G$.

We next turn this membership oracle for $\phi_G^*$ 
into a validity oracle for $\phi_G^*$. Towards this end, we collect some
geometric prerequisites about $\phi_G$ and $\phi_G^*$. 
Recall the standard norm inequalities 
$\|y\|_2\leq \|y\|_1\leq \sqrt{|E(G)|}\|y\|_2$ for 
$y:E(G)\rightarrow\mathbb{Q}$. Since $\|y\|_2\leq\frac{1}{n}$ implies
$\|y\|_1\leq n\|y\|_2\leq 1$, from \eqref{eq:phi-polyhedron} 
we observe that $\phi_G$ contains the 2-norm ball of 
radius $\frac{1}{n}$ centered at the origin; thus, the polar $\phi_G^*$ is 
a polytope contained in the 2-norm ball of radius $n$ centered at the origin. 
We also observe that $\phi_G^*\subseteq\mathbb{Q}^{|E(G)|}_{\geq 0}$; 
indeed, any $z:E(G)\rightarrow\mathbb{Q}$ with $z(\{u,v\})<0$ for
some $\{u,v\}\in E(G)$ has $z^\top y_0>1$ for $y_0:E(G)\rightarrow\mathbb{Q}$
with $y_0(\{u,v\})<1/z(\{u,v\})$ and zero elsewhere, 
with $y_0\in\phi_G$. 
Consider an arbitrary $z:E(G)\rightarrow\mathbb{Q}_{\geq 0}$ 
with $\|z\|_2\leq\frac{1}{n}$. From \eqref{eq:phi-polyhedron}
it follows that for all $y\in\phi_G$ we have
$z^\top y\leq \|z\|_1\leq n\|z\|_2\leq 1$ and thus
$z\in\phi_G^*$ by \eqref{eq:polar}. It follows that
$\phi_G^*$ contains the nonnegative part of the 2-norm ball of
radius $\frac{1}{n}$ centered at the origin and that $z_{\{u,v\}}\geq 0$
is a face of $\phi_G^*$ for each $\{u,v\}\in E(G)$; so we
have in particular that $\phi_G^*$ contains the
2-norm ball of radius $\frac{1}{2n^2}$ centered at 
$(\frac{1}{2n^2},\frac{1}{2n^2},\cdots,\frac{1}{2n^2})$.
Furthermore, it is immediate that the cone of $\phi_G$ contains 
$\mathbb{Q}^{|E(G)|}_{\leq 0}$, and we claim that equality holds;
indeed, suppose that $y:E(G)\rightarrow\mathbb{Q}$ belongs to the 
cone of $\phi_G$ with $y(\{u,v\})>0$ for some $\{u,v\}\in E(G)$,
and obtain a contradiction to $z\in\phi_G^*$ for $z(\{u,v\})=\frac{1}{n}$
and zero elsewhere. 

From the halfspace-representation~\eqref{eq:phi-polyhedron} we 
observe that the facet-complexity~\cite[Definition~6.2.2]{GrotschelLS1993} 
of $\phi_G\subseteq\mathbb{Q}^{|E(G)|}$ is bounded by $O(n^2)$, 
implying~\cite[Lemma~6.2.3]{GrotschelLS1993} that the 
vertex-complexity~\cite[Definition~6.2.2]{GrotschelLS1993} of
$\phi_G$ is bounded by $O(n^6)$.
Since the polar $\phi_G^*$ admits halfspace-representation with each 
inequality obtained from the vertex-representation%
\footnote{We recall that any nonempty 
polyhedron $\pi\subseteq\mathbb{Q}^d$ admits 
representation as the Minkowski sum $\pi=\conv(P)+\cone(R)$ of a convex hull 
of a finite set of points $P\subseteq\mathbb{Q}^d$ and the convex cone of
a finite set of rays $R\subseteq\mathbb{Q}^d$ so that the polar
$\pi^*=\{z\in\mathbb{Q}^d:y^\top z\leq 1\text{ for all $y\in\pi$}\}\subseteq\mathbb{Q}^d$ of $\pi$ has the halfspace-representation 
$\pi^*=\{z\in\mathbb{Q}^d:
p^\top z\leq 1\text{ for all $p\in P$ and } 
r^\top z\leq 0\text{ for all $r\in R$ }\}$, 
cf.~\cite[Sect.~0.1]{GrotschelLS1993}.}{}
of $\phi_G$, it immediately follows that 
the facet-complexity of $\phi_G^*$ is bounded by $O(n^6)$.

Since we know an explicit positive-radius ball inside $\phi_G^*$
as well as an explicit positive-radius ball containing $\phi_G^*$,
and the facet complexity of $\phi_G^*$ is bounded by a polynomial in $n$,
the membership oracle for $\phi_G^*$ can be turned into a weak violation
and thus weak validity oracle for 
$\phi_G^*$~\cite[Theorem~4.3.2]{GrotschelLS1993}, which in turn can be
turned into a strong validity oracle for 
$\phi_G^*$~\cite[Theorem~6.3.2]{GrotschelLS1993}, which 
in turn gives a membership oracle for the polar of the polar $\phi_G^{**}$. 
Since $\phi_G$ is closed and convex as well as contains the origin, 
we have $\phi_G^{**}=\phi_G$~\cite[Sect.~0.1]{GrotschelLS1993}) and 
thus a membership oracle for $\phi_G$. 

Finally, we will use the membership oracle for $\phi_G$ to solve 
the $k$-clique problem on $G$. 
To decide whether $G$ has a clique of size $k$, set up the constant function 
$y^{(\kappa)}:E(G)\rightarrow\mathbb{Q}$ with 
$y_{\{u,v\}}^{(\kappa)}=\kappa$ for all $\{u,v\}\in E(G)$ for some 
constant $\kappa>0$ 
yet to be determined. Observe that we can use the membership oracle 
for $\phi_G$ to test whether $y^{(\kappa)}\in\phi_G$. 
For every clique $T\subseteq V(G)$ of $G$, we observe that
$\kappa\binom{|T|}{2}>|T|-1$ if and only if $|T|>\frac{2}{\kappa}$. 
Thus, using the membership oracle for $\phi_G$, 
we can decide whether $G$ has a clique 
of size $k$ by setting $\kappa=\frac{2}{k-1}$; 
indeed, from \eqref{eq:phi-polyhedron} we have that 
the constant-function query $y^{(\kappa)}\notin\phi_G$ if and only if 
$G$ has a clique of size of at least $k$, equivalently a $k$-clique.
\end{proof}

\paragraph*{Acknowledgements.}

Heikki Mannila is supported by the Technology Industries of Finland Centennial Foundation. Chandra Kanta Mohapatra is supported by Helsinki Institute for Information Technology (HIIT) postdoctoral fellowship.

%%%%%%%%%%%%%%%%%%%%%%%%%%%%%%%%%%%%%%%%%%%%%%%%%%%%%%%%%%%%%%%% References %%%

\bibliography{paper}

%%%%%%%%%%%%%%%%%%%%%%%%%%%%%%%%%%%%%%%%%%%%%%%%%%%%%%%%%%%%%%%%%% Appendix %%%

\appendix

\section{Omitted proofs}

\label{sect:proofs}

This appendix restates and proves the claims whose proofs were omitted
in the paper proper.

\vpgeom*

\begin{proof}
The vertices $(\zeta^{(\mathcal{F},T)},e^{(T)})\in\{0,1\}^{|\mathcal{F}|+2^n}$
for $T\subseteq[n]$ of $\vp^{(\mathcal{F})}$
are easily solved from the systems of equations obtained
by leaving one constraint out from the system~\eqref{eq:param-zeta} in all
possible ways, and observing that the resulting system gives a vertex 
satisfying all the constraints in all other cases except when 
the leftmost constraint in~\eqref{eq:param-zeta} is left out.
The $2^n$ vertices are affinely independent so the dimension of 
$\vp^{(\mathcal{F})}$ 
is at least $2^n-1$. The dimension of $\vp^{(\mathcal{F})}$ 
is at most $2^n-1$ because the 
$1+|\mathcal{F}|$ hyperplanes in~\eqref{eq:param-zeta} are
independent. 
\end{proof}

\cpgeom*

\begin{proof}
For Item~1, the zero-one-valued vertices property is
immediate from~\eqref{eq:set-inclusion-indicator} and
Proposition~\ref{prop:venn}(1).
From~\eqref{eq:set-inclusion-indicator} we observe that the vectors
$\zeta^{(\mathcal{F},T)}$ for $T\in\{\emptyset\}\cup\mathcal{F}$ are affinely
independent; indeed, $\zeta^{(\mathcal{F},\emptyset)}$ is the zero vector,
and the vectors $\zeta^{(\mathcal{F},T)}$ for $T\in\mathcal{F}$ 
listed in a linearization of the partial order $(\mathcal{F},\subseteq)$ 
form a triangular matrix with nonzero diagonal. Thus, 
$\cp^{(\mathcal{F})}$ is full-dimensional. For Item~2, 
observe that for all $T,U\subseteq[n]$ we have 
from \eqref{eq:set-inclusion-indicator} that 
$\zeta^{(\mathcal{F},T)}\neq \zeta^{(\mathcal{F},U)}$ if and only
if there exists an $S\in\mathcal{F}$ with 
$\iv{S\subseteq T}\neq\iv{S\subseteq U}$; this inequation holds
if and only if the sets $T$ and $U$ lie in distinct cells of the
set partition $\{\uparrow\!S,2^{[n]}\setminus\!\uparrow\!S\}$. The cells
of the meet 
$\bigwedge_{S\in\mathcal{F}}\{\uparrow\!S,2^{[n]}\setminus\!\uparrow\!S\}$
are thus in bijective correspondence with the vectors in the set
$\{\zeta^{(\mathcal{F},T)}:T\subseteq[n]\}$.
Observe also that $T$ and $U$ lie in distinct cells of the
set partition $\{\uparrow\!\{k\},2^{[n]}\setminus\!\uparrow\!\{k\}\}$
if and only if exactly one of $T$ and $U$ contains $k\in[n]$. Thus,
$\bigwedge_{S\in\mathcal{F}}\{\uparrow\!S,2^{[n]}\setminus\!\uparrow\!S\}$
is the discrete partition with $2^n$ cells when $\{k\}\in\mathcal{F}$ holds
for all $k\in[n]$. Conversely, suppose that $\{k\}\notin\mathcal{F}$ for
some $k\in [n]$, implying that 
$\zeta^{(\mathcal{F},\emptyset)}=\zeta^{(\mathcal{F},\{k\})}$ and thus
$|\{\zeta^{(\mathcal{F},T)}:T\subseteq[n]\}|<2^n$.
For Item~3, recall our discussion of Hailperin's linear program in 
the introduction.
\end{proof}

\upgeom*

\begin{proof}
For Item~1, the zero-one-valued vertices property is
immediate from~\eqref{eq:set-inclusion-indicator}
and Proposition~\ref{prop:venn}(1).
Since $\up^{(\mathcal{F})}$ coordinate-restricts to 
$\cp^{(\mathcal{F})}$, we have that $\up^{(\mathcal{F})}$ has 
dimension at least $|\mathcal{F}|$ by full-dimensionality of 
$\cp^{(\mathcal{F})}$, cf.~Proposition~\ref{prop:correlation}(1). 
It suffices to show that $\up^{(\mathcal{F})}$ has 
dimension $|\mathcal{F}|+1$ when 
$\mathcal{F}\neq 2^{[n]}\setminus\{\emptyset\}$; indeed, 
Item~2 shows that the dimension is $|\mathcal{F}|$ when 
$\mathcal{F}=2^{[n]}\setminus\{\emptyset\}$. So let us assume
$\mathcal{F}\neq 2^{[n]}\setminus\{\emptyset\}$. We split into two cases.
First, suppose that there exists an $k\in [n]$ with $\{k\}\notin\mathcal{F}$.
From~\eqref{eq:set-inclusion-indicator} we observe that the vectors
$(\zeta^{(\mathcal{F},T)},\iv{T=\emptyset})$ for 
$T\in\{\emptyset,\{k\}\}\cup\mathcal{F}$ are affinely
independent; indeed, $(\zeta^{(\mathcal{F},\{k\})},0)$ is the zero vector,
and the vectors $(\zeta^{(\mathcal{F},T)},\iv{T=\emptyset})$ 
for $T\in\{\emptyset\}\cup\mathcal{F}$ listed in a reverse linearization 
of the partial order $(\{\emptyset\}\cup\mathcal{F},\subseteq)$, placing
the empty set last, form a triangular matrix with nonzero diagonal. Thus, 
$\up^{(\mathcal{F})}$ is full-dimensional.
Second, suppose that $\{k\}\in\mathcal{F}$ holds for all $k\in[n]$.
Since $\mathcal{F}\neq 2^{[n]}\setminus\{\emptyset\}$, there thus exists 
a set $\mathcal{F}\not\ni U\subseteq [n]$ with $|U|\geq 2$ such that all 
nonempty proper subsets of $U$ are in $\mathcal{F}$.
We show that the vectors $(\zeta^{(\mathcal{F},T)},\iv{T=\emptyset})$ for 
$T\in\{\emptyset,U\}\cup\mathcal{F}$ are affinely independent.
Linearize the partial order 
$(\{\emptyset,U\}\cup\mathcal{F},\subseteq)$ so that the sets
in $2^U\subseteq \{\emptyset,U\}\cup\mathcal{F}$ come before all 
the other sets. Place the vectors 
$(\zeta^{(\mathcal{F},T)},\iv{T=\emptyset})$ into the columns of a 
matrix from left to right in this linearization order 
for $T\in\{\emptyset,U\}\cup\mathcal{F}$ so that $T=\emptyset$ gives 
the leftmost column. We can index the rows of this matrix, again in 
linearization order, by $S\in\{\emptyset\}\cup\mathcal{F}$, so that 
row $S=\emptyset$ has
the entry $\iv{T=\emptyset}$ in column $T$, and row $S\in\mathcal{F}$
has the entry $\iv{S\subseteq T}$ in column $T$ 
by~\eqref{eq:set-inclusion-indicator}. The matrix has shape 
$(|\mathcal{F}|+1)\times (|\mathcal{F}|+2)$. Now subtract the leftmost
column ($T=\emptyset$) from all the other columns, 
delete the leftmost column, and negate
the first row. It now suffices to show that 
the resulting $(|\mathcal{F}|+1)\times (|\mathcal{F}|+1)$
matrix has full rank to conclude affine independence. 
Recalling the linearization order, we observe that the matrix is upper
triangular with nonzero diagonal, apart from the top-left 
$(2^{|U|}-1)\times(2^{|U|}-1)$ submatrix, with rows indexed by
$S\in 2^U\setminus\{U\}$, and columns indexed by
$T\in 2^U\setminus\{\emptyset\}$, with entries $\iv{S\subseteq T}$. 
This top-left submatrix has full rank; indeed, the submatrix
equals, assuming lexicographic order is used to linearize $2^U$, the
$\bar Z$-matrix obtained from the base case $d=1$, $Z=M=(1)$, and $f=g=(1)$ 
by applying the construction in the following immediate
claim $|U|-1$ times, yielding $\bar Z$ of shape 
$(2^{|U|}-1)\times (2^{|U|}-1)$. 

\begin{claim}
Let $Z$ and $M$ be mutually inverse $d\times d$ matrices and 
let $f$ and $g$ be $d\times 1$ vectors such that $\gamma=g^\top M f\neq 0$. 
Define the $(2d+1)\times(2d+1)$ matrices and $(2d+1)\times 1$ vectors
\[
\bar Z=
\left[\begin{array}{c|c|c}
Z&f&Z\\\hline
g^\top&0&g^\top\\\hline
0&f&Z\\
\end{array}\right],\ 
\bar M=
\left[\begin{array}{c|c|c}
M&0&-M\\\hline
\gamma^{-1} g^\top M&-\gamma^{-1}&0\\\hline
-\gamma^{-1} Mfg^\top M&\gamma^{-1} Mf&M\\
\end{array}\right],\ 
\bar f=
\left[\begin{array}{@{}c@{}}
f\\\hline
0\\\hline
0
\end{array}\right],\ 
\bar g=
\left[\begin{array}{@{}c@{}}
0\\\hline
0\\\hline
g
\end{array}\right]\,.
\]
Then, $\bar Z$ and $\bar M$ 
are mutually inverse and $\bar\gamma=\bar g^\top \bar M \bar f=-\gamma\neq 0$. 
\end{claim}

For Item~2, from the proof of Item~1 above we observe that 
$\up^{(\mathcal{F})}$ is full-dimensional unless 
$\mathcal{F}=2^{[n]}\setminus\{\emptyset\}$. 
When $\mathcal{F}=2^{[n]}\setminus\{\emptyset\}$, that is, when we know
all intersection probabilities, codimension $1$ with
the stated hyperplane follows from 
the inclusion-exclusion formula~\eqref{eq:inclusion-exclusion}.
For Item~3, recall the proof of Proposition~\ref{prop:correlation}(2);
observe that the cell of the set partition
$P=\bigwedge_{S\in\mathcal{F}}\,\{\uparrow\!S,2^{[n]}\setminus\!\uparrow\!S\}$
that contains the empty set equals $\{\emptyset\}$ if and only if 
$\{k\}\in\mathcal{F}$ for all $k\in [n]$; otherwise, the meet
$P\wedge\{\{\emptyset\},2^{[n]}\setminus\{\emptyset\}\}$ splits this cell
into exactly two cells. We have that 
$|P\wedge\{\{\emptyset\},2^{[n]}\setminus\{\emptyset\}\}|$ is the
number of vertices in $\up^{(\mathcal{F})}$. 
For Item~4, recall our discussion of Hailperin's linear program in 
the introduction.
\end{proof}

\unionmember*

\begin{proof}
We reduce from the fractional chromatic number problem.
Let $G$ be a graph with $n$ vertices given as input and 
assume the setting in the proof of Theorem~\ref{thm:minmax}(1) 
as regards the instance $(b,\mathcal{F})$ of minimum union probability 
constructed in that proof. Let us write $[u_\text{min},u_\text{max}]$ 
for the union probability interval of the instance $(b,\mathcal{F})$, and 
recall that $(b,x_\emptyset)\in\up^{(\mathcal{F})}$ if and only if 
$u=1-x_\emptyset\in[u_\text{min},u_\text{max}]$. 
By construction, the fractional
chromatic number $\chi_f(G)$ satisfisfies $u_\text{min}=\chi_f(G)/n$,
and from \eqref{eq:min-union-feasible} we have $u_\text{max}=1$, where 
both $u_\text{min}$ and $u_\text{max}$ as well as the vector $b$ have 
rational encoding length bounded by a polynomial in $n$. 

Now suppose we have available an oracle that given as input 
the instance $(b,\mathcal{F})$ and a rational $x_\emptyset$ decides whether 
$(b,x_\emptyset)\in\up^{(\mathcal{F})}$; that is, whether 
$u=1-x_\emptyset\in [u_\text{min},u_\text{max}]=[\chi_f(G)/n,1]$.
Using the oracle a polynomial number of times in $n$, starting from $u=1/2$ 
we can in time polynomial in $n$ use binary search to find a rational 
number $u$ with $|u-u_\text{min}|<\frac{1}{2D^2}$, where $D$ is an upper bound 
for the denominator of $u_\text{min}$; in particular, we can assume that $D$ 
has encoding length bounded by a polynomial in $n$. 
This rational $u$ thus enables us to compute 
the exact value $u_\text{min}=\chi_f(G)/n$ from $u$ in time polynomial 
in $n$ with continued fractions~\cite[Theorem 5.1.9]{GrotschelLS1993}, and 
thus solve the fractional chromatic number of $G$ in time and number of 
oracle calls bounded by a polynomial in $n$.
\end{proof}

\end{document}